\documentclass[copyright,creativecommons]{eptcs}
\providecommand{\event}{FICS 2015} 
\usepackage{breakurl}             
\usepackage{nnf}

\title{A Type-Directed Negation Elimination}

\author{Etienne Lozes\institute{LSV, ENS Cachan \& CNRS}\email{lozes@lsv.ens-cachan.fr}}

\def\titlerunning{A Type-Directed Negation Elimination}
\def\authorrunning{E. Lozes}
\begin{document}
\maketitle

\begin{abstract}
In the modal $\mu$-calculus, a formula is well-formed
if each recursive variable occurs underneath an even number of 
negations. By means of De Morgan's laws, it is easy to transform
any well-formed formula $\varphi$ into an equivalent 
formula without negations -- 
the negation
normal form of $\varphi$.
Moreover, if $\varphi$ is of size $n$, the negation normal form of $\varphi$ is
of the same size $\mathcal O(n)$. The full modal 
$\mu$-calculus and the negation normal form fragment are thus 
equally expressive 
and concise.

In this paper we extend this result to the higher-order modal fixed point 
logic (HFL),
an extension of the modal $\mu$-calculus with
higher-order recursive predicate transformers.
We present a  procedure 
that converts a formula of size $n$
into an equivalent formula 
without negations of size $\mathcal O(n^2)$ in the worst
case and $\mathcal O(n)$ when the number of variables
of the formula is fixed. 
\end{abstract}

\section{Introduction}

Negation normal forms are commonplace in many logical formalisms. To quote only
two examples, in first-order
logic, negation normal form is required by Skolemization,
a procedure that distinguishes between existential and universal quantifiers;
in the modal $\mu$-calculus, the negation normal form ensures 
the existence of the fixed points. More generally, the negation normal form
helps identifying the polarities~\cite{phdlaurent} of 
the subformulas of a given formula; for instance, in the modal $\mu$-calculus,
a formula in negation normal form syntactically describes the schema
of a parity game.

Converting a formula in a formula without negations -- or with negations
at the atoms only -- is usually easy.
By means of De Morgan's laws, negations can be ``pushed to the 
leaves'' of the formula. For the modal $\mu$-calculus without
propositional variables, this process completely eliminates negations,
because well-formed formulas are formulas where recursive variables 
occur underneath an even number of negations.
Moreover, in the modal $\mu$-calculus, 
if $\varphi$ is of size $n$, the negation normal form of $\varphi$ is
of the same size $\mathcal O(n)$.

The higher-order fixed point modal logic (HFL)~\cite{ViswanathanV04} is the 
higher-order extension of the modal $\mu$-calculus. In HFL, formulas
denote either predicates, or (higher-order) predicate transformers, 
each being possibly defined
recursively as (higher-order) fixed points. 
Since HFL was introduced, it was never suggested that negation
could be eliminated from the logic. On the contrary, Viswananthan and 
Viswanathan~\cite{ViswanathanV04} motivated 
HFL with an example expressing a form of rely
guarantee that uses negation, and they strove to make sure that
HFL formulas are correctly restricted so that fixed points always exist.
Negation normal forms in HFL would however be interesting:
they would
simplify the design of two-player games for HFL 
model-checking~\cite{Bruse14}, they 
could help defining a local model-checking algorithms for HFL, they
might help to define the alternation depth of a HFL formula, etc. 

We show that HFL actually admits negation elimination, and that like for
the modal $\mu$-calculus, every HFL formula can be converted into
a formula in negation normal form. 
The negation elimination procedure
is more involved due to higher-orderness. As a witness of this
increased complexity, our negation elimination procedure 
has a worst-case quadratic blow-up in the size of the formula, whereas
for the $\mu$-calculus the negation normal form is of linear size in the 
original formula.

\paragraph{Related Work}
Other examples
of higher-order recursive objects are
the higher-order pushdown 
automata~\cite{Maslov76,Cachat03},
or the higher-order recursion schemes 
(HORS)~\cite{Damm198295,HagueMOS08,CarayolS12,SalvatiW14}. 
Whereas the decidability of HFL model-checking against finite transition 
systems 
is rather simple, it took more time to understand the decidability of 
HORS model-checking 
against the ordinary (order 0) modal $\mu$-calculus. This situation
actually benefited to HORS: the intense
research on HORS produced several optimized algorithms and implementations
of HORS model-checking~\cite{BroadbentCHS13,FujimaIK13,TeraoK14}, 
whereas HFL model-checking remains a rather theoretical and unexplored topic.
HORS can be thought as recursive formulas with no boolean connectives and
least fixed points everywhere.
On the opposite, HFL allows any kinds
of boolean connectives, and in particular a form of
``higher-order alternation''. 

\paragraph{Outline} We recall the definition of HFL and all useful background
about it in Section~2. In Section~3, we sketch the ideas
driving our negation elimination and introduce the notion
of monotonization, a correspondence
between arbitrary functions and monotone ones that is at the core
of our negation elimination procedure.
We formally define the negation elimination procedure in Section~4, and make
some concluding remarks in Section~5.

\section{The Higher-Order Modal Fixed Point Logic}

We assume an infinite set $\Var = \{X,Y,Z,\ldots\}$
of variables, and a finite set 
$\Sigma=\{a,b,\dots\}$ of labels. Formulas 
$\varphi,\psi$, 
of the Higher-Order Modal Fixed Point Logic (HFL) 
are defined by the following grammar
$$
\form,\formbis \enspace ::= 
\top\mid\form\vee\formbis\mid\neg\form\mid\may{a}\form
\mid X\mid \lambda X^{\typ,v}.\ \form \mid \form\ \formbis \mid \mu X^{\typ}.\ \form 
$$
where a type $\tau$ is either the ground type $\grtype$ or an arrow
type $\sigma^v\to \tau$, and the \emph{variance} $v$ is either
$+$ (monotone), or $-$ (antitone), or $0$ (unrestricted).
For instance, $\tau_1=(\grtype^-\to\grtype)^+\to(\grtype^0\to\grtype)$ is 
a type, and $\varphi_1=$
$
\lambda F^{\grtype^-\to\grtype,+}.\ \lambda Y^{\grtype,0}.\ \mu Z^{\grtype}.\ (F\ \neg Y
)\vee \may{a} (Z\vee \neg Y)
$
is a formula. 
The sets $\mathsf{fv}(\varphi)$ and $\mathsf{bv}(\varphi)$
of free and bound variables of $\varphi$ are defined
as expected: $\mathsf{fv}(X)=\{X\}$, $\mathsf{bv}(X)=\emptyset$, 
$\mathsf{fv}(\lambda X.\ \varphi)=\mathsf{fv}(\mu X.\ \varphi)=\mathsf{fv}(\varphi)\setminus\{X\}$,
$\mathsf{bv}(\lambda X.\ \varphi)=\mathsf{bv}(\mu X.\ \varphi)=\mathsf{bv}(\varphi)\cup \{X\}$, etc.
A formula is \emph{closed} if $\mathsf{fv}(\varphi)=\emptyset$.
For simplicity, we restrict our attention to formulas $\varphi$
\emph{without variable masking}, i.e. such that for every
subformula $\lambda X.\ \psi$ (resp. $\mu X.\ \psi$), it holds that
$X\not\in \mathsf{bv}(\psi)$.

Another example is the formula
$\varphi_2=$ $(\lambda F^{\grtype^-\to\grtype,+}.\ \mu X^{\grtype}.\ F\ X)\ (\lambda Y^{\grtype,-}.\ \neg Y)$. This formula can be $\beta$-reduced to the modal 
$\mu$-calculus formula
$\varphi_2'=\mu X^{\grtype}.\ \neg X$, which does not have a fixed point
semantics. Avoiding ill-formed HFL formulas such as $\varphi_2$ cannot just
rely on counting the number of negations between $\mu X$ and the occurence
of $X$, it should also take into account function applications and the context
of a subformula.

A type judgement is a tuple $\Gamma\vdash \form:\typ$, where $\Gamma$ is
a set of assumptions of the form $X^v:\typ$. The typing 
environment $\neg\Gamma$ is the one in which every assumption $X^{v}:\tau$ is
replaced with $X^{-v}:\tau$, where $-+=-$, $--=+$, and $-0=0$.
A formula $\form$
is well-typed and has type $\tau$ if the type judgement $\vdash \form:\tau$
is derivable from the rules defined in Fig.~\ref{fig:typing}.
Intuitively, the type judgement 
$X_1^{v_1}:\tau_1,\dots,X_n^{v_n}:\tau_n\vdash \varphi:\tau$ is derivable if
asssuming that $X_i$ has type $\tau_i$, it may be infered that
$\varphi$ has type $\tau$ and that $\varphi$, viewed as a function
of $X_i$, has variance $v_i$.
For instance, $\vdash \varphi_1:\tau_1$, where $\varphi_1$ and $\tau_1$ are 
the formula and the type we defined above, but $\varphi_2$ cannot be typed, 
even with different type annotations.

\begin{figure}
\begin{mathpar}
\trule{ }{\judg{\Gamma}{\top:\grtype}}

\trule{\judg{\Gamma}{\form:\typ}\\\judg{\Gamma}{\formbis:\typ}}
{\judg{\Gamma}{\form\vee\formbis:\typ}}

\trule{\judg{{}\neg\Gamma}{\form:\typ}}{\judg{\Gamma}{\neg\form:\typ}}

\trule{\judg{\Gamma}{\form:\grtype}}{\judg{\Gamma}{\may{a}\form:\grtype}}

\trule{v\in\{+,0\}}{\judg{\Gamma\;,\;X^v: \typ}{X:\typ}}

\trule{\judg{\Gamma,X^v:\typbis}{\form:\typ}}
{\judg{\Gamma}{\lambda X^{v,\typbis}.\ \form :\typbis^v\to\typ}}

\trule{\judg{\Gamma,X^{+}:\typ}{\form:\typ}}
{\judg{\Gamma}{\mu X^\typ.\ \form :\typ}}

\trule{\judg{\Gamma}{\form:\typbis^{+}\to\typ}\\\judg{\Gamma}{\formbis:\typbis}}
{\judg{\Gamma}{\form\ \formbis :\typ}}

\trule{\judg{\Gamma}{\form:\typbis^{-}\to\typ}\\\judg{\neg\Gamma}{\formbis:\typbis}}
{\judg{\Gamma}{\form\ \formbis :\typ}}

\trule{\judg{\Gamma}{\form:\typbis^{0}\to\typ}\\
\judg{\Gamma}{\formbis:\typbis}\\\judg{\neg\Gamma}{\formbis:\typbis}
}
{\judg{\Gamma}{\form\ \formbis :\typ}}
\end{mathpar}
\caption{The type system of HFL.}
\label{fig:typing}
\end{figure}

\begin{proposition}\label{prop:type-uniqueness}\cite{ViswanathanV04}
If~~$\Gamma\vdash\varphi:\tau$ and $\Gamma\vdash\varphi:\tau'$ are derivable, 
then $\tau = \tau'$, and the two derivations coincide.
\end{proposition}

If $\varphi$ is a well-typed closed formula and $\psi$ is a subformula of 
$\varphi$, we write $\typede{\psi}{\varphi}$ for the type of $\psi$ in
(the type derivation of) $\varphi$.

A labeled transition system (LTS) 
is a tuple $\mathcal T=(S,\delta)$ where $S$ is a set
of states and $\delta\subseteq S\times \Sigma\times S$ is a transition
relation. For every type $\tau$ and every LTS $\mathcal T=(S,\delta)$, 
the complete Boolean ring $\mathcal T\sem{\tau}$ of interpretations of 
closed formulas of 
type $\tau$ is defined by induction on $\tau$: $\mathcal T\sem{\grtype}=2 ^S$,
and $\mathcal T\sem{\sigma^v\to \tau}$ is the complete Boolean ring 
of all total functions
$f:\mathcal T \sem{\sigma}\to \mathcal T\sem{\tau}$ that have variance $v$, where
all Boolean operations on functions are understood pointwise.
Note that since $\mathcal T\sem{\tau}$ is a complete Boolean ring, it is
also a complete lattice, and any monotone function 
$f:\mathcal T\sem{\tau}\to\mathcal T\sem{\tau}$ admits a unique least fixed 
point.

A $\mathcal T$-valuation $\rho$
is a function that sends every variable of type $\tau$
to some element of $\mathcal T \sem{\tau}$. More precisely,
we say that $\rho$ is well-typed according to some typing environment
$\Gamma$, which we write $\rho\models \Gamma$, if $\rho(X)\in \mathcal T \sem{\tau}$
for every $X^v:\tau$ in $\Gamma$.
The semantics $\mathcal T\sem{\Gamma\vdash \varphi:\tau}$ of a derivable
typing judgement is a function that associates to every $\rho\models \Gamma$
an interpretation $\mathcal T\sem{\Gamma\vdash \varphi:\tau}(\rho)$ 
in $\mathcal T\sem{\tau}$; this interpretation
is defined as expected by induction on the derivation tree 
(see~\cite{ViswanathanV04} for details). For a well-typed
closed formula $\varphi$ of type $\grtype$, a LTS $\mathcal T=(S,\delta)$ and
a state $s\in S$, We write $s \models_{\mathcal T} \varphi$ 
if $s\in\mathcal T\sem{\vdash\varphi:\grtype}$.

\begin{example}
Let $\tau_3=(\grtype^+\to \grtype)^+\to\grtype^+\to\grtype$ and $\varphi_3=$
$$
\big(\mu F^{\tau_3}.\ \lambda G^{\grtype^+\to\grtype},
X^{\grtype}.\ (G\ X)\vee\big(F\ (\lambda Y^{\grtype}.G\ (G\ Y))\ X\big)\big)\quad
(\lambda Z^{\grtype}.\ \may{a} Z)\quad \may{b}\top.
$$
Then $s\models \varphi_3$ iff there is $n\geq 0$ such that there is a path
of the form $a^{2^n}b$ starting at $s$. Since $\{a^{2^n}b\mid n\geq0\}$ is
not a regular language, the property expressed by $\varphi_3$
cannot be expressed in the modal $\mu$-calculus.
\end{example}

\begin{proposition}\cite{ViswanathanV04}
Let $\mathcal T=(S,\delta)$ be a LTS and let $s,s'\in S$ be two bisimilar states
of $\mathcal T$. Then for any closed formula $\varphi$ of type $\grtype$,
$s\models_{\mathcal T}\varphi$ iff $s'\models_{\mathcal T}\varphi$.
\end{proposition}

We assume the standard notations $\wedge$, $\must{a}$ and $\nu X.\ (.)$
for the conjunction, the necessity modality, and the greatest fixed point, 
defined as the duals of $\vee$, $\may{a}$ and $\mu X.\ (.)$ respectively.

\begin{definition}[Negation Normal Form]
A HFL formula is in negation normal form if it is derivable from the grammar
$$
\form,\formbis \enspace ::= 
\top\mid \bot\mid\form\vee\formbis\mid \form\wedge \formbis \mid\may{a}\varphi\mid\must{a}\varphi\mid
X\mid \lambda X^{\typbis}.\form \mid \form\ \formbis \mid \mu X^{\typ}.\form \mid \nu X^{\typ}.\form
$$
where the $\tau$ are \emph{monotone} types, i.e. types
where all variances are equal to $+$.
\end{definition}

Note that since all variances are $+$, we omit them when writting formulas
in negation normal form. 

We say that two formulas $\varphi,\psi$ are 
equivalent, $\varphi\equiv\psi$, if for every type environment $\Gamma$, 
for every LTS $\mathcal T$, for all type $\tau$, the judgement 
$\Gamma\vdash\varphi:\tau$
is derivable iff $\Gamma\vdash\psi:\tau$ is, and in that case
$\mathcal T\sem{\Gamma\vdash\varphi:\tau}=\mathcal T\sem{\Gamma\vdash\psi:\tau}$. 

\paragraph{Model-Checking}

We briefly recall the results known about the data complexity of
HFL model-checking (see also the results of Lange~\emph{et al}
on the combined complexity~\cite{AxelssonLS07} or the
descriptive complexity~\cite{LangeL14} of HFL and extensions).

Note that if $\mathcal T=(S,\delta)$ is a finite LTS, then for all type $\tau$,
the Boolean ring $\mathcal T\sem{\tau}$ is a finite set, and every element of
$\mathcal T\sem{\tau}$ can be represented \emph{in extension}. Moreover,
the least fixed point of a monotone function 
$f:\mathcal T\sem{\tau}\to\mathcal T\sem{\tau}$ can be computed by
iterating $f$ at most $n$ times, where $n$ is the size 
of the finite boolean ring $\mathcal T\sem{\tau}$.

The order $\mathsf{ord}(\tau)$ of a type $\tau$ is defined as
$\mathsf{ord}(\grtype)=0$ and $\mathsf{ord}(\sigma^v\to\tau)=\max(\mathsf{ord}(\tau),1+\mathsf{ord}(\sigma))$. We write $\HFL{k}$ to denote the set of
closed HFL formulas $\varphi$ 
of type $\grtype$ such that all type annotations in $\varphi$ are of order at 
most $k$. For every fixed $\varphi\in \HFL{k}$, 
we call $\mathsf{MC}(\varphi)$ the problem of deciding, given a LTS 
$\mathcal T$ and a state $s$ of $\mathcal T$, wether 
$s\models_{\mathcal T}\varphi$.
\begin{theorem}\cite{AxelssonLS07}
For every $k\geq 1$, for every $\varphi\in\HFL{k}$, the problem 
$\mathsf{MC}(\varphi)$ is in $k$-EXPTIME, and there is a $\psi_k\in\HFL{k}$ such
that $\mathsf{MC}(\psi_k)$ is $k$-EXPTIME hard.
\end{theorem}

\section{Monotonization\label{sec:monotonization}}

In order to define a negation elimination procedure, the first idea is 
probably to reason like in the modal $\mu$-calculus, 
and try to ``push the negations to the leaves''. Indeed, 
there are De Morgan laws for all logical connectives, including abstraction
and application, since
\begin{mathpar}
\neg (\varphi\ \psi)\quad\equiv\quad(\neg \varphi)\ \psi

\and \mbox{and} \and
\neg (\lambda X^{v,\tau}.\psi)\quad \equiv\quad \lambda X^{-v,\tau}.\neg \psi.
\end{mathpar}
In the modal $\mu$-calculus, this idea is enough, because the 
``negation counting'' criterion ensures that each pushed negation eventually
reaches another 
negation and both anihilate. This does not happen for HFL.
Consider for instance the formula $\varphi_4=$
$$
\big(\mu X^{\grtype^0\to\grtype}.\ \lambda Y^{\grtype,0}.\
(\neg Y)\vee \big(X\ (\may{a} Y)\big)\big)\quad \top.
$$
The negation already is at the leaf, but $\varphi_4$ is not
in negation normal form. By fixed point unfolding, one can check
that $\varphi_4$ is equivalent to the infinite disjunct 
$\bigvee_{n\geq 0} \must{a}^n{\bot}$, and thus could be expressed
by $\mu X^{\grtype}.\must{a}X$. The generalization of this
strategy for arbitrary formulas would be interesting,
but it is unclear to us how it would be defined.

We follow another approach: we do not try to unfold fixed points
nor to apply $\beta$-reductions during negation elimination, but we stick
to the structure of the formula. In particular, in our approach 
a subformula denoting a function $f$ is mapped to a subformula
denoting a function $f'$ in the negation normal form. Note that
even if $f$ is not monotone, $f'$ must be monotone since it is a subformula of
a formula in negation normal form. We call $f'$ a \emph{monotonization}
of $f$.

\paragraph{Examples}

Before we formaly define monotonization, we illustrate its principles
on some examples. 

First, consider again the above formula $\varphi_4$. This formula 
contains the function $\lambda Y^{\grtype,0}.\ (\neg Y)\vee \big(X\ (\may{a} Y)\big)$. This function is unrestricted (neither monotone nor antitone).
 The monotonization of this function will be the function 
$\lambda Y^{\grtype,+},\bar Y^{\grtype,+}. \bar{Y}\vee \big(X\ (\may{a} Y)\big)$.
To obtain this function, a duplicate $\bar{Y}$ of $Y$ is introduced, 
and is used in place of $\neg Y$. Finally,
the formula $\varphi_4'=$
$$
\big(\mu X^{\grtype\to\grtype\to\grtype}.\ \lambda Y^{\grtype},\bar Y^{\grtype}.\
\bar Y\vee \big(X\ (\may{a} Y)\ (\must{a}\bar Y)\big)\big)\quad \top \qquad \bot
$$
can be used as a negation normal form of $\varphi_4$. Note that
the parameter $\top$ that was passed to the recursive function 
in $\varphi_4$ is duplicated in $\varphi_4'$, with one duplicate
that has been negated (the $\bot$ formula).

More generally, whenever a function is of type $\sigma^{0}\to\tau$, we transform
it into a function of type $\sigma_{t}^+\to\sigma_t^+\to\tau_t$
that takes two arguments of type $\sigma_t$ (the translation of 
$\sigma$). Later, when this function is applied, we make sure that
its argument is duplicated, one time positively, the other negatively.

Duplicating arguments might cause an exponential blow-up.
For instance, for the formula $\varphi_5=$
$$
(\lambda X^{\grtype}.\ X\vee\may{a}\neg X)\quad \big((\lambda Y^{\grtype,0}.\ Y\vee\may{b}\neg Y)\ \top\big)
$$
if we duplicated arguments naively, we could get the formula $\varphi_5'=$
$$
(\lambda X^{\grtype},\bar X^{\grtype}.\ X\vee\may{a}\bar X)\quad 
\big((\lambda Y^{\grtype},\bar Y^{\grtype}.\ Y\vee\may{b}\bar Y)\ \top\ \bot\big)
\quad \big((\lambda Y^{\grtype},\bar Y^{\grtype}.\ \bar Y\wedge\must{b} Y)\ \top\ \bot\big)
$$
where the original $\top$ formula has been duplicated. 
If it occurred underneath $n+2$ applications of an unrestricted
function, we would have $2^n$ copies of $\top$.
We will come back to this problem in Section~\ref{sec:negelim}.

Let us now observe how monotonization works for functions that are 
antitone. In general, if $f$ is an antitone function, both 
the ``negation at the caller'' $f_1(x)=\neg f(x)$ and 
the ``negation at the callee'' $f_2(x)=f(\neg x)$ are two monotone functions
that faithfully represent $f$. Actually, 
both of them might be needed by our negation
elimination procedure.

Consider the formula $\varphi_6=$
$$
(\lambda F^{\grtype^-\to\grtype,+}.\mu X^{\grtype}.F\ (\neg X))\quad (\lambda Y^{\grtype,-}.\neg \may{a} Y).
$$
In order to compute the negation normal form of $\varphi_6$,
we may represent $\lambda Y^{\grtype,-}.\neg \may{a} Y$ by its
``negation at the callee'', yielding the formula
$\varphi_6'=$
$$
(\lambda F^{\grtype\to\grtype}.\mu X^{\grtype}.F\ X)\quad (\lambda \bar Y^{\grtype}. \must{a} \bar Y).
$$

Conversely, consider the formula
$\varphi_7=$
$$
(\lambda F^{\grtype^-\to\grtype,-}.\ \mu X^{\grtype}.\ (\neg F)\ X)\quad (\lambda Y^{\grtype,-}.\ \neg \may{a} Y).
$$
The only difference with $\varphi_6$ is that the negation is now
in front of $F$ instead of $X$. In that case, ``negation at the callee''
does not help eliminating negations. But
``negation at the caller'' does, and yields the
negation normal form $\varphi_7'=$
$$
(\lambda \bar F ^{\grtype\to\grtype}.\ \mu X^{\grtype}.\ \bar F\ X)\quad (\lambda Y^{\grtype}.\ \may{a} Y).
$$

These examples suggest a negation elimination that proceeds 
along possibly different strategies in the case of
an application $\varphi\ \psi$, depending on the semantics of $\varphi$
and $\psi$. In the next section, we explain how the strategy is determined 
by the type of $\varphi$. For now, we focus on making more formal our notion
of monotonization.

\paragraph{Monotonization Relations}

We saw that our negation elimination bases on the ability to
faithfully represent a predicate transformer $\varphi$ by a 
monotone predicate transformer $\psi$; in this case,
we will say that $\psi$ is a \emph{monotonization} of $\varphi$. 
We now aim at defining formally this notion. More precisely, we aim at
defining
the relation $\triangleleft$ such that $\varphi \triangleleft \psi$ holds
if $\psi$ is a monotonization of $\varphi$.

\begin{figure}
$$
\begin{array}{cc}
\begin{array}{ll}
\exp{}{\grtype} & = \grtype
\\
\exp{}{\typ^+\to\typbis} & = \exp{}{\typ}^+\to\exp{}{\typbis}
\\
\exp{}{\typ^-\to\typbis} & = \exp{}{\typ}^+\to\exp{}{\typbis}
\\
\exp{}{\typ^0\to\typbis} & = \exp{}{\typ}^+\to\exp{}{\typ}^+\to\exp{}{\typbis}
\end{array}
&
\begin{array}{ll}
\exp{}{\Gamma_1,\Gamma_2} & = \exp{}{\Gamma_1},\exp{}{\Gamma_2}
\\
\exp{}{X^+:\tau} & = X^+:\exp{}{\tau}
\\
\exp{}{X^-:\tau} & = \bar X^+:\exp{}{\tau}
\\
\exp{}{X^0:\tau} & = X^+:\exp{}{\tau},\bar X^+:\exp{}{\tau}
\end{array}
\end{array}
$$
\caption{\label{fig:monotonization}Expansion of types and typing environments towards monotonization.}
\end{figure}

First of all, $\triangleleft$ relates a formula of type $\tau$ to
a formula of type $\exp{}{\tau}$ as defined in 
Fig.~\ref{fig:monotonization}: the number of arguments of 
$\varphi$ is duplicated
if $\varphi$ is unrestricted, otherwise it remains the same, and of course 
$\psi$ is monotone in all of its arguments.

In Fig.~\ref{fig:monotonization}, we also associate to every typing environment
$\Gamma$ the typing environment $\exp{}{\Gamma}$ with all variances
set to $+$, obtained after renaming all variables with variance $-$ in 
their bared version, and duplicating all variables with variance $0$.
In the remainder, we always implicitly
assume that we translate formulas and typing environments that do not initially
contain bared variables.

The relation $\triangleleft$ is then defined coinductively, in a similar way as
logical relations for the $\lambda$-calculus. Let $\mathbin R$
be a binary relation among typing judgements of the form $\Gamma\vdash \varphi:\tau$.
The relation $\mathbin R$ is
well-typed if $(\Gamma\vdash \varphi:\tau)\mathbin R(\Gamma'\vdash \varphi':\tau')$
implies $\Gamma'=\exp{}{\Gamma}$ and $\tau'=\exp{}{\tau}$. When $\mathbin R$
is well typed, we write $\varphi R_{\Gamma,\tau} \varphi'$ instead of
$(\Gamma\vdash \varphi:\tau)\mathbin R(\Gamma'\vdash \varphi':\tau')$.

\begin{definition}
A binary relation $\mathbin R$ among typing judgements
is a \emph{monotonization relation} if it is well-typed, and for all
formulas $\varphi,\varphi'$, for all $\Gamma,\tau$ such that $\varphi\mathbin R_{\Gamma,\tau}\varphi'$, 
\begin{enumerate}
\item if $\varphi,\varphi'$ are closed and $\tau=\grtype$, then $\varphi\equiv \varphi'$;
\item if $\Gamma=\Gamma',X^+:\sigma$, then 
$(\lambda X^{\sigma,+}.\ \varphi) \mathbin R_{\Gamma',\sigma^+\to\tau} (\lambda X^{\exp{}{\sigma},+}.\ \varphi')$;
\item if $\Gamma=\Gamma',X^-:\sigma$, then 
$(\lambda X^{\sigma,-}.\ \varphi) \mathbin R_{\Gamma',\sigma^-\to\tau} (\lambda \bar X^{\exp{}{\sigma},+}.\ \varphi')$;
\item if $\Gamma=\Gamma',X^0:\sigma$, then 
$(\lambda X^{\sigma,0}.\ \varphi) \mathbin R_{\Gamma',\sigma^0\to\tau} (\lambda X^{\exp{}{\sigma},+},\bar X^{\exp{}{\sigma},+}.\ \varphi')$;
\item 
if $\tau=\sigma^+\to\upsilon$, then for all $\psi,\psi'$ such that
$\psi\mathbin R_{\Gamma,\sigma} \psi'$, $(\varphi\ \psi)\mathbin{R}_{\Gamma,\upsilon} (\varphi'\ \psi')$;
\item 
if $\tau=\sigma^-\to\upsilon$, then for all $\psi,\psi',\psi''$ such that
$\psi\mathbin R_{\Gamma,\sigma} \psi'$ and $\psi'\equiv \neg \psi''$, 
$(\varphi\ \psi)\mathbin{R}_{\Gamma,\upsilon} (\varphi'\ \psi'')$;
\item 
if $\tau=\sigma^0\to\upsilon$, then for all $\psi,\psi',\psi''$ such that
$\psi\mathbin R_{\Gamma,\sigma} \psi'$ and $\psi'\equiv \neg \psi''$, 
$(\varphi\ \psi)\mathbin{R}_{\Gamma,\upsilon} (\varphi'\ \psi'\ \psi'')$.
\end{enumerate}
\end{definition}

If $(R_i)_{i\in I}$ is a family of monotonization relation, then so is
$\bigcup_{i\in I} R_i$; we write $\triangleleft$ for the largest 
monotonization relation.

\begin{example}
Consider $\varphi=(\lambda X^{\grtype,-}.\ \neg X)$. Then $\varphi\triangleleft_{\grtype^-\to\grtype} (\lambda \bar X^{\grtype,+}.\ \bar X)$. Consider also
$\psi=(\lambda X^{\grtype,0}.\ X\wedge \neg X)$. Then
$\psi\triangleleft (\lambda X^{\grtype,+},\bar X^{\grtype,+}.\ \bot)$ and
$\psi\triangleleft (\lambda X^{\grtype,+},\bar X^{\grtype,+}.\  X\wedge \bar X)$.
\end{example}

\section{Negation Elimination\label{sec:negelim}}

Our negation elimination procedure proceeds in two steps: first,
a formula $\varphi$ is translated into a formula
$\trl{+}{\varphi}$ that denotes the monotonization of $\varphi$; 
then, $\trl{+}{\varphi}$ is concisely 
represented in order to avoid an exponential blow-up.

The transformation $\trl{+}{.}$ is presented in Figure~\ref{fig:neg-elim}.
The transformation proceeds by structural induction on the formula,
and is defined as a mutual induction with the companion transformation
$\trl{-}{.}$. Whenever a negation is encountered, it is eliminated and
the dual transformation is used. As a consequence, wether $\trl{+}{.}$
or $\trl{-}{.}$ should be used for a given subformula
depends on the polarity~\cite{phdlaurent} of this subformula.

\begin{lemma}
\label{lem:neg-elim}
Let $\varphi$ be a fixed closed formula of type $\grtype$. For every
subformula $\psi$ of $\varphi$, let $\trl{+}{\psi}$ and $\trl{-}{\psi}$
be defined as
in Figure~\ref{fig:neg-elim}, and let 
$\Gamma\vdash\psi:\tau$ be the type judgement associated to $\psi$
in the type derivation of $\varphi$. Then the following statements hold.
\begin{enumerate}
\item $\exp{}{\Gamma}\vdash\trl{+}{\psi}:\exp{}{\tau}$ and 
$\exp{}{\neg\Gamma}\vdash\trl{-}{\psi}:\exp{}{\tau}$.
\item $\psi\mathbin{\triangleleft_{\Gamma,\tau}}\trl{+}{\psi}$
 and $\psi\mathbin{\triangleleft_{\Gamma,\tau}} \neg \trl{-}{\psi}$.
\end{enumerate}
\end{lemma}
\begin{proof}
By induction on $\psi$. We only detail the point 1 in the case of
$\psi=\psi_1\ \psi_2$ with
$\typede{\psi_1}{\varphi}=\sigma^{-}\to \tau$. Let us assume the two statements
hold for $\psi_1$ and $\psi_2$ by induction hypothesis.
Let $\Gamma$ be such that $\Gamma\vdash \psi:\tau$, 
$\Gamma\vdash \psi_1:\sigma^{-}\to\tau$, and $\neg \Gamma\vdash \psi_2:\sigma$.
By induction
hypothesis, the judgements 
$\exp{}{\Gamma}\vdash \trl{+}{\psi_1}:\exp{}{\sigma^-\to\tau}$ and
$\exp{}{\neg\neg\Gamma}\vdash \trl{-}{\psi_2}:\exp{}{\sigma}$ are derivable.
Since $\exp{}{\sigma^-\to\tau}=\exp{}{\sigma}^+\to\exp{}{\tau}$ and $\neg\neg \Gamma=\Gamma$, the typing rule for function application in the monotone case
of Fig.~\ref{fig:typing} yields
$\exp{}{\Gamma}\vdash\trl{+}{\psi_1}\ \trl{-}{\psi_2}:\exp{}{\tau}$,
which shows statement 1 for $\trl{+}{.}$. The case for $\trl{-}{.}$ is 
similar. 
\end{proof}

\begin{figure}
$$
\begin{array}{cc}
\begin{array}{r@{~=~}l}
\trl{+}{\top} & \top \\ \trl{-}{\top} & \bot \\
\trl{+}{X} & X \\ \trl{-}{X} & \bar X \\
\trl{v}{\neg\psi} & \trl{-v}{\psi}\\
\trl{+}{\may{a}\psi} & \may{a} \trl{+}{\psi}\\
\trl{-}{\may{a}\psi} & \must{a} \trl{-}{\psi}\\
\end{array}
& 
\begin{array}{r@{~=~}l}
\trl{+}{\psi_1\vee\psi_2} & \trl{+}{\psi_1}\vee\trl{+}{\psi_2}\\
\trl{-}{\psi_1\vee\psi_2} &  \trl{-}{\psi_1}\wedge\trl{-}{\psi_2}\\
\trl{v}{\lambda X^{\tau,+}.\ \psi} & \lambda X^{\exp{}{\tau}}.\ \trl{v}{\psi}\\
\trl{v}{\lambda X^{\tau,-}.\ \psi} & \lambda \bar X^{\exp{}{\tau}}.\ \trl{v}{\psi}\\
\trl{v}{\lambda X^{\tau,0}.\ \psi} & \lambda X^{\exp{}{\tau}},\bar X^{\exp{}{\tau}}.\ \trl{v}{\psi} \\
\trl{+}{\mu X^{\tau}.\ \psi} & \mu X^{\exp{}{\tau}}.\ \trl{+}{\psi}\\
\trl{-}{\mu X^{\tau}.\ \psi} & \nu \bar X^{\exp{}{\tau}}.\ \trl{-}{\psi}
\end{array}
\end{array}
$$
$$
\trl{v}{\psi_1\ \psi_2}~~=~~\left\{\begin{array}{ll}
\trl{v}{\psi_1}\ \trl{+}{\psi_2} & \mbox{if }\typede{\psi_1}{\varphi}=\sigma^+\to\eta \\
\trl{v}{\psi_1}\ \trl{-}{\psi_2} & \mbox{if }\typede{\psi_1}{\varphi}=\sigma^-\to\eta \\
\trl{v}{\psi_1}\ \trl{+}{\psi_2}\ \trl{-}{\psi_2} & \mbox{if }\typede{\psi_1}{\varphi}=\sigma^0\to\eta \\
\end{array}\right.
$$
\caption{\label{fig:neg-elim}Type-Directed Negation Elimination}
\end{figure}

\begin{corollary}\label{coro:neg-elim}
If $\varphi$ is a closed formula of type $\grtype$, then
$\varphi\equiv \trl{+}{\varphi}$ and $\trl{+}{\varphi}$ is in negation
normal form.
\end{corollary}

As observed in Section~\ref{sec:monotonization}, the duplication of
the arguments in the case $v=0$ of the monotonization of $\varphi \psi$
may cause an exponential blow-up in the size of the formula. However,
this blow-up does not happen if we allow some sharing
of identical subformulas.

Let $\varphi$ be a fixed closed formula. We say that two subformulas
$\psi_1$ and $\psi_2$ of $\varphi$ are identical if they are
syntactically equivalent and if moreover they have the same
type and are in a same typing context, i.e. if the type derivation
of $\varphi$ goes through the judgements $\Gamma_i\vdash \psi_i:\tau_i$
for syntactically equivalent $\Gamma_i$ and $\tau_i$.
For instance, in the formula
$$
(\lambda X^{\grtype\to\grtype}.\ X)\quad \big((\lambda X^{(\grtype\to\grtype)\to
(\grtype\to\grtype)}.\ X)\quad \big((\lambda Y^{\grtype\to\grtype}.Y)\ \top\big)\big)
$$
any two distinct subformulas are not identical (including the subformulas
restricted to $X$).
We call \emph{dag size}
of $\varphi$ 
the number of non-identical subformulas of $\varphi$.

\begin{lemma}
\label{lem:share}
There is a logspace computable function $\mathsf{share}(.)$
that associates to every closed formula $\varphi$ of dag size $n$ a
closed
formula $\mathsf{share}(\varphi)$ of tree size 
$\mathcal O(n\cdot |\mathsf{vars}(\varphi)|)$
such that $\varphi\equiv\mathsf{share}(\varphi)$.
\end{lemma}

\begin{proof}
Let $\varphi$ be fixed, and
let $\varphi_1\dots,\varphi_n$
be an enumeration of all subformulas of $\varphi$ such
that if $\varphi_{i}$ is a strict subformula of $\varphi_{j}$, then $i<j$.
In particular, we must have $\varphi=\varphi_n$.
Pick some fresh variables $X_1,X_2,\dots, X_n\in \Var$
and let $\upsilon_i=\typede{\varphi_i}{\varphi}$.
For every $i=1,\dots n$, let
$Y_1,\sigma_1,v_1,\dots Y_k,\sigma_k,v_k$ be a fixed enumeration
of the free variables of $\varphi_i$, their types and their variances, and let
$\lambda_i(\psi)=\lambda Y_1^{\sigma_1,v_1},\dots,Y_k^{\sigma_k,v_k}.\ \psi$
and $@_i(\psi)=\psi\ Y_1\ \dots\ Y_k$. Finally, let 
$\tau_i=\sigma_1^{v_1}\to\dots\sigma_k^{v_k}\to\upsilon_i$.
For every subformula $\psi$ of $\varphi$, let $\|\psi\|$ be defined by case 
analysis on the first logical connective of $\psi$:
\begin{itemize}
\item 
if $\psi=\varphi_i=\eta Y^{\sigma}.\ \varphi_j$, where 
$\eta\in\{\lambda,\mu,\nu\}$,
then 
$\|\psi\|= \lambda_i\big(\eta Y^{\sigma}.\ @_j(X_j)\big)$;
\item 
if $\psi=\varphi_i=\varphi_j\oplus\varphi_k$, where $\oplus\in\{\vee,\wedge,\mbox{application}\}$,
then 
$\|\psi\|=\lambda_i\big(@_j(X_j)\oplus @_k(X_k)\big)$;
\item 
if $\psi=\varphi_i=\spadesuit\varphi_j$, where $\spadesuit\in\{\neg,\may{a},\must{a}\}$,
then $\|\psi\|=\lambda_i\big(\spadesuit (@_j(X_j))\big)$;
\item otherwise $\|\varphi_i\|=\lambda_i(\varphi_i)$.
\end{itemize}
Finally, 
let $\mathsf{share}(\varphi)=$
$
\mathbf{let}~X_{1}^{\tau_1}=\|\varphi_1\|~\mathbf{in}$ 
$\mathbf{let}~X_{2}^{\tau_2}=\|\varphi_2\|~\mathbf{in}$
$\dots$ 
$\mathbf{let}~X_{n-1}^{\tau_{n-1}}=\|\varphi_{n-1}\|~\mathbf{in}~\|\varphi_n\|$
where 
$\mathbf{let}~X^{\tau}=\psi~\mathbf{in}~\psi'$ is a 
macro for $(\lambda X^{\tau}.\ \psi')\ \psi$. Then $\mathsf{share}(\varphi)$
has the desired properties.
\end{proof}

\begin{theorem}
\label{thm:nnf}
There is a logspace-computable function $\nnf(.)$ that associates to every
closed HFL formula $\varphi$ (without variable masking) of type 
$\grtype$ a closed
formula $\nnf(\varphi)$
such that 
\begin{enumerate}
\item  $\varphi\equiv\nnf(\varphi)$,
\item $\nnf(\varphi)$ is in negation normal form, and 
\item $|\nnf(\varphi)|=\mathcal O(|\varphi|\cdot|\mathsf{vars}(\varphi)|)$,
\end{enumerate} where $|\psi|$ denotes the size of the tree representation
of $\psi$ (i.e. the number of symbols in $\psi$), and 
$\mathsf{vars}(\varphi)=\mathsf{fv}(\varphi)\cup\mathsf{bv}(\varphi)$ is 
the set of variables that occur in $\varphi$. 
\end{theorem}
\begin{proof}
Let $\nnf(\varphi)=\mathsf{share}(\trl{+}{\varphi})$. This function is 
logspace computable (
$\trl{+}{\varphi}$ can be computed ``on-the-fly'') and $\nnf(\varphi)$
is of size $\mathcal O(|\varphi|\cdot|\mathsf{vars}(\varphi)|)$
by Figure~\ref{fig:neg-elim} and Lemma~\ref{lem:share}. 
The formula $\trl{+}{\varphi}$
is in negation normal form, and $\mathsf{share}(.)$ does not introduce new
negations, so $\nnf(\varphi)$ is in negation normal form. Looking back
at Figure~\ref{fig:neg-elim}, it can be checked that its dag size is linear
in the dag size of $\varphi$, so the tree size of $\nnf(\varphi)$ is
linear in the tree size of $\varphi$. Moreover,
$\nnf(\varphi)\equiv\trl{+}{\varphi}$ by Lemma~\ref{lem:share},
and $\trl{+}{\varphi}\equiv \varphi$ by Corollary~\ref{coro:neg-elim}.
\end{proof}

\section{Conclusion}

We have considered the higher-order modal fixed point 
logic~\cite{ViswanathanV04} 
(HFL) and its fragment 
without negations, and we have shown that both formalisms 
are equally expressive. More
precisely, we have defined a procedure for transforming any closed HFL formula 
$\varphi$ denoting a state predicate 
into an equivalent formula $\nnf(\varphi)$ without negations of
size $\mathcal O(|\varphi|\cdot|\mathsf{vars}(\varphi)|)$. 
The procedure works in two phases: in a first phase, 
a transformation
we called \emph{monotonization} eliminates all negations
and represents arbitrary functions of type $\tau \to \sigma$
by functions of type $\tau\to\tau\to\sigma$ by distinguishing
positive and negative usage of the function parameter. The price
to pay for this transformation is an exponential blow-up in
the size of the formula. If the formula is represented as a circuit, however,
the blow-up is only linear. The second phase of our negation
elimination procedure thus consists in implementing
the sharing of common subformulas using higher-orderness. 
Thanks to this second phase, our procedure yields
a negation-free formula $\nnf(\varphi)$ of size 
$\mathcal O(\mathsf{size}(\varphi)\cdot |\mathsf{vars}(\varphi)|)$, 
hence quadratic in the worst case in the size of the original formula $\varphi$.

\paragraph{Typed versus Untyped Negation Elimination}
Our monotonization procedure is \emph{type-directed}: the monotonization of
$\varphi\ \psi$ depends on the variance of $\varphi$, that is statically
determined by looking at the type of $\varphi$. One might wonder if
we could give a negation elimination that would not be type-directed.
A way
to approach this question is to consider an untyped conservative extension 
of the logic where
we do not have to care about the existence of the fixed points -- for instance,
one might want to interprete $\mu X.\varphi(X)$ as the inflationary 
``fixed point''~\cite{DawarGK04}. We believe that we could adapt our monotonization
procedure to this setting, and it would indeed become a bit simpler:
we could always monotonize $\varphi\ \psi$ ``pessimistically'', as 
if $\varphi$ were neither a 
monotone nor an antitone function. For instance,
the formula $\mu X.(\lambda Y. Y)\ X$ would be translated into
$\mu X. (\lambda Y,\bar Y. Y)\ X\ \neg X$. 

In our typed setting,
it is crucial to use the type-directed monotonization we developed, because
monotizing pessimistically might yield ill-typed formulas. 
In an untyped setting, a pessimistic monotonization 
is possible, but it yields less concise
formulas, and it looses the desirable property that 
$\nnf(\nnf(\varphi))=\nnf(\varphi)$. 

So types, and more precisely variances, seem quite unavoidable. However,
strictly speaking, the monotonization we introduced is 
\emph{variance-directed}, and not really type-directed. In particular, 
our monotonization might be extended to the untyped setting, relying on 
some other static analysis than types to determine 
the variances of all functional subformulas.

\paragraph{Sharing and Quadratic Blow-Up}
The idea of sharing subterms of a $\lambda$-term
is reminiscent to implementations of $\lambda$-terms based on
hash-consing~\cite{Filliatre-hash,Goubault-hash} and to compilations of
the $\lambda$ calculus into interaction nets~\cite{Lamping90,Mackie98,Gonthier92thegeometry}.
We showed how sharing can be represented directly in the $\lambda$-calculus,
whereas hash-consing and interaction nets are concerned with representing 
sharing either in memory or as a circuit. We compile
typed $\lambda$-terms into typed $\lambda$-terms; a consequence
is that we do not manage to share 
subterms that are syntactically identical but have either different types
or are typed using different type assumptions for their free variables.
This is another difference with hash consing and interaction nets, where
syntactic equality is enough to allow sharing subterms.
It might be the case that we could allow more sharing if we did not compile
into a simply typed $\lambda$-calculus but in a ML-like language
with polymorphic types.

An interesting issue is the quadratic blow-up of our implementation
of ``$\lambda$-circuits''.
One might wonder wether a more succinct negation elimination is possible,
in particular a negation elimination with linear blow-up.
To answer this problem, it would help to answer the following
simpler problem: \emph{given a $\lambda$-term $t$ with
$n$ syntactically distinct subterms, is there an effectively
computable $\lambda$-term $t'$ of size $\mathcal O(n)$ such
that $t=_{\beta\eta} t'$}? We leave that problem for future work.

\bibliographystyle{eptcs}
\bibliography{refs}

\end{document}